\documentclass[a4wide]{amsart}

\usepackage{amsmath}  
\usepackage{amsfonts}
\usepackage{times}
\usepackage{helvet}

\usepackage{array}
\usepackage{graphicx}
\usepackage{color}
\usepackage{mathtools}
\usepackage{bbm}
\usepackage[swedish,english]{babel}

\usepackage[latin1]{inputenc}
\usepackage[T1]{fontenc}

\newtheorem{theorem}{Theorem}[section]
\newtheorem{lemma}[theorem]{Lemma}

\newtheorem*{theorem*}{Theorem}
\theoremstyle{remark}
\newtheorem{remark}[theorem]{Remark}
\newtheorem{definition}[theorem]{Definition}

\newcommand{\R}{\mathbb{R}}

\newcommand{\s}{\mathrm{S}}
\newcommand{\bigoh}{\mathcal{O}}

\newcommand{\one}{\mathbbm{1} }
\newcommand{\norm}[1]{\|#1\|}

\newcommand{\di }{\mathrm{d}}

\newcommand{\fnm}{F_{\tiny N M}}
\newcommand{\barfnm}{\bar{F}_{\tiny N M}}
\newcommand{\fnmnm}{f_{N M}^{(n m)}}
\newcommand{\fnmettnoll}{f_{N M}^{(1 0)}}
\newcommand{\fnmnnoll}{f_{N M}^{(N 0)}}
\newcommand{\fnmlillannoll}{f_{N M}^{(n 0)}}

\title{The BGK equation as the limit of an $N$-particle system}

\newcommand{\eact}{{N+M-|V|^2}}

\makeatletter
\renewcommand*\env@matrix[1][\arraystretch]{%
  \edef\arraystretch{#1}%
  \hskip -\arraycolsep
  \let\@ifnextchar\new@ifnextchar
  \array{*\c@MaxMatrixCols c}}
\makeatother

\begin{document}

\author{Dawan Mustafa$^{(1)}$ and Bernt Wennberg$^{(2,3)}$}


\address{$^{(1)}$ University of Borås, SE50190 Borås,  Sweden}

\address{$^{(2)}$ Department of Mathematical Sciences \\ Chalmers
  University of Technology \\ SE41296 G\"oteborg \\ Sweden}

\address{$^{(3)}$ Department of Mathematical Sciences \\ University of
  Gothenburg \\ SE41296 G\"oteborg \\ Sweden}

\email{dawan.mustafa@hb.se, wennberg@chalmers.se}

\begin{abstract}
The spatially homogeneous BGK equation is obtained as the limit of a model of a
many particle system, similar to Mark Kac's charicature of the
spatially homogeneous Boltzmann equation.
\end{abstract}


\maketitle 

\section{Introduction}
\label{sec_intro}

The BGK equation is named after P.L. Bhatnagar, E.P. Gross, and
M. Krook, who first presented it in an influential paper published in
1954~\cite{BhatnagarGrossKrook1954}. In its original form it is

\begin{align}
  \label{eq:1}
   \frac{\partial f}{\partial t} +
     v\cdot\nabla f -\frac{e \mathbf{E}}{m} \frac{\partial f}{\partial
     v} & = -\frac{n}{\sigma} f + \frac{n^2}{\sigma} \Phi\,.
\end{align}
Here $f=f(x,v,t)$ gives the number density of particles in phase-space
$(x,v)\in \R^3\times\R^3$. The constant  $\sigma>0$ controls collision
rate of particles, and $\Phi=\Phi_{q,T}$ is the Maxwellian distribution
\begin{align}
  \label{eq:2}
  \Phi(x,v,t)= \Phi_{q,T} &= \frac{m}{\left(2\pi k T(x,t)\right)^{3/2}}
      \exp\left(-\frac{m}{2kT(x,t)}\left(v-q(x,t)\right)^2\right)\;,
\end{align}
where  $n(x,t)$, $q(x,t)$, and $T(x,t)$ represent the local number
density, mean velocity and temperature respectively:
\begin{align}
  \label{eq:3}
n(x,t) &= \int f(v,x,t)\,\di v\,,\nonumber \\
q(x,t) &= \frac{1}{n(x,t)} \int v f(v,x,t)\,\di v\,,\\
\frac{3kT(x,v)}{m} &= \frac{1}{n(x,t)}\int (v-q(x,t))^2 f(v,x,t)\di v\,.
                     \nonumber 
\end{align}
The same kind of equation was formulated independently by
Welander~\cite{Welander1954}.  
In~\cite{BhatnagarGrossKrook1954}, one considers charged particles,
and $\mathbf{E}$ is the electric field computed from the particle
density. It is a model of the kinetic Boltzmann  equation with the
purpose of providing a numerically tractable model, while retaining
the most important aspects of the original 
Boltzmann equation: conservation of mass, momentum and energy,
convergence to a unique equilibrium state, monotonicity of entropy,
etc. And while easier from a computational point of view, it is
considerably more difficult to analyse mathematically, and most
theoretical results concerning 
existence and uniqueness of solutions to the BGK eqution actually hold
for a modified version where the right hand side is replaced by
\begin{align}
  \label{eq:4}
   -\frac{1}{\sigma} f + \frac{n}{\sigma} \Phi\,,
\end{align}
{\em i.e.} where the collision frequency is
constant~\cite{Perthame1989,PerthamePulvirenti1993}. There are also
results concerning solutions close to a global equilibrium, which hold
also  for density and temperature dependent collision
frequencies~\cite{Yun2010,Yun2015}. There is a rather  
large litterature concerning various aspects of the BGK-equation
dealing, for example,  with methods for numerical treatment of
 rarefied gases (some recent examples
are~\cite{Groppi_etal_2016,AsadzadehKazemiMokhtari2014,XuanXu2013}),
their fluid dynamical limits (see for example
\cite{SaintRaymond2002b,Colosqui2010,ChenOrszagStaroselsky2007}),
or models accounting for polyatomic 
gases or  mixtures of different gases (for example
in~\cite{AndriesAokiPerthame2002,GroppiSpiga2013,BisiCaceres2016}), to
give a few examples. A paper attempting to find a well-motivated
approximation of the collision frequenecy $1/\sigma$ can be found
in~\cite{Simons1972}.

The BGK equation is a fenomenological equation in the sense that it is derived
explicitly to satisfy certain physical properties of a dilute gas, but
until very recently there are very few published works attempting to justify
the equation directly from the dynamics of an $N$-particle
system. This is in contrast with the Boltzmann equation, for which
there is now a rigorous derivation starting from the Liouville
equation for hard sphere dynamcis, or for short range potentials,  of
an $N$-particle system~\cite{Gallagher_etal2013,Lanford1975}.

 The BGK equation without electric field can be interpreted  
as a model of a large system of particles, where each particle moves
independently with its own velocity. The  velocity
jumps at exponentially distributed intervals but remains constant in
between the jumps. The jump rate is proportional to the local density
of the gas, and after the jump the 
particle velocity is a normally distributed random variable,
independent of the initial velocity, but with mean and variance 
determined by the local temperature and mean velocity of the gas. 
One can make a similar interpretation of the Boltzmann equation for
hard spheres, but with two important differences. First,
in~(\ref{eq:1}), the  collision rate only depends on the local
density, and not on the velocity of the particles. In fact,
equation~(\ref{eq:1}) corresponds to a system of so-called Maxwellian
molecules and not to hard spheres.  The second, and
more important, difference lies in the distribution of velocities
of particles after a jump. The  Boltzmann equation for Maxwellian
molecules, which in similar notation is
\begin{align}
  \label{eq:5}
   \frac{\partial f}{\partial t} +
     v\cdot\nabla f & = -\frac{n}{\sigma} f + \frac{1}{\sigma}Q^+(f,f)\,,
\end{align}
represents a process in which the velocity of a particle after the
jump is given by the outcome of a random collision with a second
particle drawn from the distribution with density $v\mapsto f(x,v,t) /
\int_{\R^3} f(x,w,t)\,\di w$. The solutions to~(\ref{eq:5}) converge
to a Maxwellian distribution when $t\rightarrow\infty$, or
equivalently, when the average number of velocity jumps that one
particle has made, goes to  
infinity. In the BGK model, the velocity of a particle has a normal
distribution after only one jump, and a particle system that
converges to a solution of the BGK model must achieve that in the
limit of infinitely many particles.

The particle system that we propose  consists of particles that can have
two states, active and passive, where only the active particles
participate in collisions with other particles. The BGK equation will
describe the evolution of passive particles in the limit of infinitely
many particles. One may think of the active particles being ions, that
interact at a high rate with each other, the passive ones being
neutrals that do not interact. On the other hand, a neutral particle and an ion
may encounter and interchange state by the transfer of an
electron, so that the result is similar to allowing the
velocity of a neutral particle to jump to a random velocity given by the
distribution of the acitve particles. And if the collision rate for
active particles is very high, then the active particles will have
time to come close to an equilibrium distribution before the next
exchange with the passive particles takes place.

At a formal level, one may actually pursue these ideas to derive a BGK
equation of the form~(\ref{eq:1}), or a hard sphere version of the
same, but to make  a completely rigorous derivation along the lines of
for example~\cite{Gallagher_etal2013} seems to be difficult~\cite{JungWennberg_inprogress}. 

An alternativ approach has been
developed in~\cite{Butta2020particle}, where an $N$-particle system is
constructed in which  the particles are given a normal
velocity distribution after a jump, with moments computed from the
empirical distributions. The authors prove rigorously that the
$N$-particle model converges to the BGK model in the limit of $N$ going
to infinity.  

Long before a rigorous result on the validity of the Boltzmann equation had
been obtained for a real particle system, Mark Kac \cite{Kac1957}
proposed a Markov 
jump process for the velocities of an $N$-particle distribution, and
proved that in the limit as $N\rightarrow\infty$, the velocity
distribution of one particle converges to the solution of a
Boltzmann-like equation for a spatially homogenous gas of Maxwellian
molecules with one-dimensional velocities. 

In this paper we  construct a Kac-type model of a system of $N$
passive and $M$ active particles, and a jump process involving
collisions between active particles and the switch between active and
passive state, as described above. We then prove that the one-particle
distribution for passive particles converges to a BGK equation of the
form
\begin{align}
  \label{eq:6}
  \partial_t f(v,t) &= \mathcal{M}(v,t) - f(v,t)
\end{align}
where $\mathcal{M}$ is the standard normal distribution in one
dimension. This limit can be obtained in a scaling
where $M/N\rightarrow 0$ when $N\rightarrow\infty$, that is, when the
fraction of active particles vanishes in the limit of infinitely many particles.

The paper is based on the results in the doctoral thesis of the first
author~\cite{Mustafa2015}. A very  similar model, with two different kinds of
particles, has been presented by Bonetto
et. al. in~\cite{Bonettoetal2014b}, and also in~\cite{Bonetto_etal2017}.
The authors are in general interested in kinetic models coupled with a
thermostat, and in the cited papers the larger set of particles ($N$ in
our paper) is  considered as a thermostat acting on the smaller set of
particles, and they prove that indeed, when $N\rightarrow\infty$ the
large $N$-particle system is a good approximation of a Gaussian thermostat.
 Related results can also be found in~\cite{Tossounian_etal2015}.

The paper is organised as follows: In Section~\ref{sec:model}, we
discuss Markov jump processes that give BGK-like equations in the
limit of infinitely many particles, and present in full detail our
final model. In Section~\ref{sec:initialsteps} we introduce some
further notation, and present the initial steps of the proof. An
important step of the proof is to show that the energy partition
between the passive and active particles in the limit is such that the
mean energy for the active particles is one. This is proven in
Section~\ref{sec:moments} by
performing very explicit calculations of certain moments of the solutions.
The proof is then concluded in Section~\ref{sec:final}.

\section{The particle system, and its limiting kinetic equation}
\label{sec:model}

We consider a particle system consisting of $N+M$ particles, where $N$
is number of \emph{passive} particles represented by
$V=(v_1,\dots,v_N)\in\R^{N}$, and $M$ is number of \emph{active} particles
represented by $W=(w_1,\dots,w_M)\in\R^{M}$. One active particle is
assumed to have the same mass as one passive
particle, and here that mass is set to $1$. The total kinetic energy,
which thefore is  $\frac{1}{2}\left( v_1^2+\cdots+ v_N^2 +w_1^2
  +\cdots + w_M^2\right)$,   is
assumed to be conserved, and therfore  the state space of the particle system
is $\s^{N+M-1}(\sqrt{N+M})$, the $N+M-1$-dimensional sphere of radius
$\sqrt{N+M}$.  Throughout the paper we also assume that $N>M$.

The dynamics of the system consist of two independent jump
processes. The first one is the so-called Kac walk, which mimics the
pairwise collisions of a rarefied gas. This process involves only the
active particles.  The second process involves a pair consisting of
one active and one passive particle, and leads to an exchange of
state: the passive particle becomes active, while retaining its
velocity, and the active particle becomes passive. In this way there
is an exchange of energy between the two sets of particles.

The Kac walk on the set of active particles is defined as follows:
\begin{itemize}
\item the jumps occur at exponentially distributed intervals with
  rate  $\lambda_2 N M$.
\item in a jump, a pair $(w_i, w_j)$ is is chosen uniformly among the
  active particles, and and $\theta\in[-\pi,\pi[$ is drawn from an
  even distribution. Then  $(w_i, w_j)\mapsto
  (w_i\cos\theta-w_j\sin\theta, w_i\sin \theta + w_j\cos\theta)$. With
  $W=(w_1,...,w_N)$, this jump is denoted $W\mapsto R_{i,j}(\theta) W$.
\end{itemize}
The jump rate for exchange between active and passive particles is
chosen as $ \lambda_1 N$, which means that the jump rate for a given
indexed passive particle $v_j$ is $\lambda_1$, independently of $N$, and
that the rate at which active particles become passive is $\lambda_1 N
M^{-1}$ per active particle. Therefore  the subsystem of active
particles on average experience $\frac{\lambda_2}{\lambda_1} M $ jumps
of Kac-type between two exchange jumps. Without loss of generality we
set is $\lambda_1=1$, and denote the parameter $\lambda_2$ simply as $\lambda$ in what follows.  

The intuitive picture is this: consider a time interval $[t_1,t_2[$,
where the end points are given by two consecutive exchange events. In
this interval the vector $V=(v_1,...,v_N)$ is unchanged, and the
vector $W=(w_1,...,w_M)$ will make on the order of
$\lambda M N (t_2-t_1)\sim \lambda M$ steps in the Kac walk. The energy
of the set of 
active particles is conserved by this process, and hence
$|W|^2=w_1^2+\ldots+w_M^2$ is constant. If $\lambda$ is
very large, the Kac walk will drive the distribution of $W$ to an
almost uniform distribution on the sphere defined by $|W(t_1)|^2$. At
$t_2$ a new exchange event takes place, when a randomly chosen active
particle becomes passive, and hence the set of passive particles will
gain a particle drawn from a distribution which is the marginal of the
uniform distribution of an $M-1$-dimensional sphere. But this marginal
distribution is close to a Gaussian when $M$ is large, and therefore,
looking only at the distribution of passive particles, this will loose
particles at exponential rate $\lambda_1=1$ per passive particle, and
gain particles drawn from a Gaussian distribution with the same rate;
this is the BGK-process for a spatially homogenous gas.

All of this can be quantified, but some notation is needed in order to
formulate a theorem. First of all we define the master equation, or
forward Kolmogorov equation, corresponding to the jump process.
Let $\fnm(V,W,t)$ be the probability density with respect to the induced
measure 
$\sigma$ \footnote{The symbol $\sigma$ is used throughout
  the paper to  
  denote the measure on the sphere $\s^{n-1}(r) $ induced from the
  Euclidian measure in $\R^n$, and therefore it is defined only in combination
  with the domain of integration.} on $\s^{N+M-1}(\sqrt{N+M})$ for the velocities of the
particles at time $t$. The time evolution of $\fnm$ is given
by the equation:
\begin{equation}\label{mastereqVW}
\frac{\partial }{\partial t}\fnm(V,W,t)=(L_{NM\lambda} +  U_{NM})\fnm(V,W,t),
\end{equation}
where
\begin{equation}\label{L}
L_{NM\lambda} \fnm(V,W)= {\textstyle \frac{2N\lambda}{M-1}}\sum_{1\leq j<k\leq M}\int_{0}^{2\pi}\left(\fnm(V,R_{jk}(\theta)W)-\fnm(V,W)\right)\ \frac{d\theta}{2\pi},
\end{equation}
and
\begin{equation}\label{U}
U_{NM}\fnm(V,W)={\textstyle \frac{1}{M}}\sum_{j=1}^{N}\sum_{k=1}^{M}\left(
  \fnm(V_{w_k}^j,W_{v_j}^k)-\fnm(V,W)\right)\,.
\end{equation}
The operator $L_{NM\lambda}$ defined in~(\ref{L}) is the generator
of the original 
Kac master equation acting on the $W$-variables, but with a factor
$\lambda N$ in front, to give the jump rate as described above. The
operator $U_{NM}$ defined in~(\ref{U}), with 
\begin{equation}
  \label{eq:Vwkjdef}
  (V_{w_k}^j, W_{v_j}^k) =(\underbrace{v_1,\dots,v_{j-1},w_k,v_{j+1},\dots,v_N}_
  {V_{w_k}^j},\underbrace{w_1,\dots,w_{k-1},v_j,w_{k+1},\dots,w_M}_{W_{v_j}^k})
\end{equation}
is the generator of the exchange process, when a passive and an active
particle exchange their state.

An essential assumption here, just like in Kac's orginal work, is that
$\fnm$ is symmetric with respect to permutations of the coordinates of
$V$ and of the coordinates of $W$. This is to say that all passive
particles are identical, and identically distributed, and that the
same holds for the active particles. Hence any choice of $n$ passive
particles is equivalent to choosing the first $n$. The following
notation will be useful:
\begin{align}
  \label{eq:VnWmdef}
  V_n&= (v_1,...,v_n)\qquad\mbox{and} \qquad V^n= (v_{n+1},...,v_N)\\
  W_m&= (w_1,...,w_m)\qquad\mbox{and} \qquad W^m= (w_{m+1},...,w_m)\,.\nonumber
\end{align}

\begin{definition}
  \label{def:1}
The $(n,m)$-marginals $\fnmnm(V_n,W_m)$ of $\fnm(V,W)$ are given 
by the equation 
\begin{equation}
\begin{split}
   \int\limits_{\mathclap{\s^{\cramped{N+M-1}}(\sqrt{N+M})}}&\quad g(V_n) h(W_m)\fnm(V,W)d
   \sigma   \\  
    &=\int_{\R^n}\int_{\R^m}g(V_n) h(W_m)\fnmnm(V_n,W_m)\
    dV_n \ dW_m, 
\end{split}
\end{equation}
where $g(V_n)$ and $h(W_m)$ are any bounded continuous functions
on $\R^n$, $\R^m$, respectively, and we assume that $\fnmnm$ has support in $\{V_n^2+W_m^2\le N+M\}$. 
\end{definition}

The objective of this paper is to prove that when
$N,M\rightarrow\infty$ the density of one passive particle,
$\fnmettnoll(v,t)$ converges to a function
$f(v,t)$ that satisfies the spatially homogeneous
BGK equation~(\ref{eq:5}). This is formulated in the following theorem:
\begin{theorem}
  \label{thm:main}
Let  $\{F_{NM}(V,W,t) \}_{N,M}$ be the solutions of a family of master
equations (\ref{mastereqVW}) with $1\le M < N <\infty$, with 
initial data $\fnm(V,W,0)$ satisfying
\begin{equation}
  \label{eq:223}
  \begin{split}
&\int\limits_{\Omega_{0,0}}|\fnm(V,W,0)|^2\ d\sigma(V,W)\leq C_{N M}^2 
\hspace{0.5cm} ,  \hspace{0.5cm}
\int\limits_{\Omega_{0,0}}\fnm(V,W,0)v_1^4\ d\sigma(V,W) <\infty,\\ 
&\mbox{and} \\
&
\frac{1}{M}\int\limits_{\Omega_{0,0}}\fnm(V,W,0)\left(\frac{1}{M}\sum_{k=1}^M
  w_k^4\right) \ d\sigma(V,W)\rightarrow 0\qquad\mbox{when}\quad
M\rightarrow\infty\,. 
\end{split}
\end{equation}
Let $M=M(N)$, $\lambda=\lambda(N)$ be such that $N/M \rightarrow
\infty$, $N/M^2 
\rightarrow 0$ and $\lambda/C_{NM} \rightarrow \infty $ when $N
\rightarrow \infty$. 
Then, for $0<t_0 \leq t \leq T$, \qquad $T<\infty$
\begin{equation}
\label{eq:224}
\lim_{N\rightarrow \infty}
  \frac{\partial}{\partial t} \int\limits_{\mathclap{v_1^2<N+M}}\quad
  \fnmettnoll(v_1,t) g(v_1)\ dv_1 
=\int_{\R}\left(\mathcal{M}(v_1)-f(v_1,t)\right)g(v_1) \ dv_1, 
\end{equation}
where
\begin{equation}
\label{eq:225}
f(v_1,t)=\lim_{N,M \rightarrow \infty}\fnmettnoll(v_1,t)\,,
\end{equation}
and where $f(v_1,t)$ solves the homogeneous BGK equation,
\begin{align}
  \frac{\partial}{\partial t} f(v_1,t) &= \mathcal{M}(v_1)-f(v_1,t)\,.
\end{align}

\end{theorem}

So, at least weakly, the one-particle marginal of the $N+M$
dimensional particle system converges to the solution of a BGK
equation, as announced in the introduction. The theorem is stated to
hold uniformly for $t\ge t_0>0$, but to achieve convergence uniformly
for all $t>0$ one must make stronger assumptions on initial data. If
the initial data are chaotic, {\em i.e.} meaning that the
many-particle marginals are close to products of functions of the
coordinates, the $L^2$-norm in the theorem grows exponentially in
$N+M$, and choosing $C_{N M} \sim c^{N+M}$ in equation~(\ref{eq:223})
gives natural class of inital data for which the theorem holds. 

\begin{remark}
  An important notion in kinetic theory is that of propagation of
  chaos, which was made precise in Kac's paper~\cite{Kac1957}. In the
  present context we would say that $\{ F_{N M} \}$ is a chaotic
  family if the marginals $\{ \fnmnm\}$ satisfy
  \begin{align}
    \lim_{N,M \rightarrow\infty} \fnmlillannoll(v_1,...,v_n) =
    \prod_{j=1}^n \lim_{N,M \rightarrow\infty} \fnmettnoll(v_j)\,,
  \end{align}
and that propagation of chaos holds if the same property holds for all
times provided it holds initially. Here only marginals with respect to
the $v$ variables are included because the limiting equation only
involves the distribution of passive particles. These do not interact
directly, but jump almost independently. When $N$ and $M$ are bounded,
some correlation is created because each jump of a passive particle
changes the distribution of the active particles, the effect of this
vanishes when the number of active particles, $M$,  increase to
infinity. A more rigorous statement can be made from the observation
that the proof of Theorem~\ref{thm:main} with very small changes shows that
equation~(\ref{eq:224}) also holds for the marignals $f_{NM}^{(2 0)}$,
whith the Maxwellian $\mathcal{M}(v_1)$ replaced by a bivariate
Maxwellian $\mathcal{M}(v_1,v_2)$, which itself factorizes.
\end{remark}

\section{Initial steps of the proof}
\label{sec:initialsteps}

To prepare for the proof of Theorem~\ref{thm:main}, we first present a
few well-known formulae concerning spheres. A first observation is that
although the $V$- and $W$-variables represent particles in different
states, they behave exactly as variables for integration over the
sphere $\Omega_{0,0}= \{|V|^2+|W|^2=N+M\}$, and therefore, for any
function $G(V,W)$, we 
have
\begin{align}
  \label{eq:16b}
  \int_{\Omega_{0,0}} G(V_{w_k}^j,W_{v_j}^k)\,d\sigma(V,W) &=
  \int_{\Omega_{0,0}} G(V,W)\,d\sigma(V,W)\,.
\end{align}

The area of an $n-1$ dimensional sphere of radius $r$,

$\s^{n-1}(r)=\left\{
  (x_1,...,x_n) \in\R^n\,\big|\, x_1^2+x_2^2\dots+x_n^2=r^2\right\}$
is given by
\begin{align}
  \label{eq:014}\\
  \nonumber
  \left|\s^{n-1}(r)\right| &=  r^{n-1} \frac{2
           \pi^{\frac{n}{2}}}{\Gamma\left(\frac{n}{2}\right)}
            \,=\, r^{n-1} \left|\s^{n-1}\right|\,.
\end{align}
For any function $f$ defined on $\s^{n-1}(r)$ one may write
\begin{align}
  \label{eq:014b}
  \int\limits_{\s^{n-1}(r)} f(x_1,...,x_n) \,d\sigma(x_1,...,x_n) &=\\
  \nonumber
  \int\limits_{x_1^2+\dots+x_k^2\le r^2}
 {\tiny \left(\frac{r^2}{r^2-x_1^2-\dots-x_k^2}\right)^{1/2}}
  &\int\limits_{\mathclap{\s^{n-k-1}\left(\sqrt{r^2-x_1^2-\dots-x_k^2}\right)}}
  f(x_1,\dots,x_n)\, d\sigma(x_{k+1},\dots,x_n)\,dx_1dx_2\dots dx_k\,.
\end{align}
Therefore the marginals defined in Definition~\ref{def:1} may be
written explicitly as
\begin{equation}\label{nmmarginal}
\fnmnm(V_n,W_m)=\Gamma_{n,m}\int\limits_{\Omega_{n,m}}\fnm(V,W)\
                       d\sigma(V^n,W^m)\,, 
\end{equation}
where we introduce the notation
\begin{align}
  \label{eq:015b}
  \Gamma_{n,m}&=
    \frac{(N+M)^{1/2}}{(N+M-|V_n|^2-|W_m|^2)^{1/2}},
                \qquad\mbox{and}
  \\  
  \Omega_{n,m}
       &= \s^{N+M-n-m-1}\left(\sqrt{N+M-|V_n|^2-|W_m|^2}\right)\,.
        \nonumber         
\end{align}

We also define the average with respect to the $W$-variables as follows:
\begin{definition}
  \label{def:2}
  Let $\fnm(V,W)\in L^1\left(\s^{N+M-1}(r)\right)$. Then
  \begin{align}
    \label{eq:def2eq}
    \barfnm(V) &= \left|\s^{M-1}\left(\sqrt{r^2-|V|^2}\right)\right|^{-1}
    \int\limits_{\s^{M-1}\left(\sqrt{r^2-|V|^2}\right)}\fnm(V,W)\,d\sigma(W)\,.
    \end{align}
\end{definition}
Finally we compute the marginal of the first coordinate of a point
chosen uniformly on an $M$-dimensional sphere of radius $\sqrt{M}$. The
uniform density is given by the constant function
$\left|\s^{M-1}\right|^{-1} M^{-(M-1)/2}$, and hence the marginal of the
first coordinate is
\begin{align}
  \label{eq:016a}
  {\mathcal M}_M(x_1) &=
   \left(\frac{M}{M-x_1^2}\right)^{1/2} 
   \int\limits_{\s^{M-2}\left(\sqrt{M-x_1^2}\right)}
                \left|\s^{M-1}\right|^{-1} M^{-(M-1)/2}
                \,d\sigma(x_2,...,x_M)
                \nonumber \\
              &= \left(\frac{M}{M-x_1^2}\right)^{1/2}
                \frac{ \left|\s^{M-2}\right|}{ \left|\s^{M-1}\right|}
                \frac{(M-x_1^2)^{(M-2)/2}}{M^{(M-1)/2}} \,=\,\frac{1}{\sqrt{M}} \frac{ \left|\s^{M-2}\right|}{ \left|\s^{M-1}\right|}
                \left(1-\frac{x_1^2}{M}\right)^{(M-3)/2}\\
                &  \underset{M\rightarrow\infty}{\longrightarrow}
                \frac{1}{\sqrt{2\pi}} e^{-x_1^2/2}\,=\,{\mathcal
                  M}(x_1)\,.\nonumber  
   \end{align}

The first step in our proof of Theorem~\ref{thm:main} is to integrate
over the $W$-variables in Equation~(\ref{mastereqVW}), to find an
evolution equation for the $V$-marginal of $F_{N,M}$. It is
\begin{equation}
  \label{eq:N0margA}
\begin{split}
   \frac{\partial}{\partial t}& \fnmnnoll(V,t) \\
   &=\Gamma_{N,0}
   \int\limits_{\Omega_{N,0}}\left(L_{NM\lambda} +
   U_{NM}\right)\fnm(V,W,t)\ d\sigma(W)  \\ 
     &=\sum_{j=1}^{N}\
     \Gamma_{N,0}\int\limits_{\Omega_{N,0}}\left(
       \fnm(V_{w_1}^j,W_{v_j}^1,t)-\fnm(V,W,t)\right) \ d\sigma(W)\,.
\end{split}
\end{equation}
The integral of $ L_{NM\lambda}\fnm(V,W,t)$ vanishes because the
generator of the Kac walk conserves mass, and we also use the symmetry
with respect to permutations of the $W$-coordinates to replace the sum
over $k$ in~(\ref{U}) with $M$ terms, all involving $w_1$. Adding and
subtracting $\barfnm(V_{w_1}^j)$ we obtain
\begin{equation}
  \label{eq:N0marg}
\begin{split}
   \frac{\partial}{\partial t}& \fnmnnoll(V,t) \\
   &=\sum_{j=1}^{N}
     \Gamma_{N,0}\int\limits_{\Omega_{N,0}}\left(\barfnm(V_{w_1}^j,t)
       -\fnm(V,W,t)\right) \ d\sigma(W) \\
          &\qquad +\sum_{j=1}^{N}
     \Gamma_{N,0}\int\limits_{\Omega_{N,0}}\left(
       \fnm(V_{w_1}^j,W_{v_j}^1,t)-\barfnm(V_{w_1}^j,t)\right) \
     d\sigma(W)\\
     &= I_1(V) + I_2(V)\,.
\end{split}
\end{equation}
We will show that by a suitable choice of $\lambda$, which is hidden
here because it only affects the Kac operator $L_{N,M,\lambda}$, the
term $I_2(V)$ vanishes in the  limit, and hence that the evolution of
$\fnmnnoll(V)$ essentially is governed by $I_1(V)$, which in turn
will reproduce the righthand side of the BGK equation in the limit
when $N,M\rightarrow\infty$.

We denote the $N$ terms of $I_1(V)$ as $I_{1,j}(V)$, so that
$I_1(V)=\sum_{j=1}^N  I_{1,j}(V)$.
For $j\ne 1$, and for any function $g\in C(\R)$, 
\begin{align*}
  \int\limits_{\mathclap{|V|^2\le N+M}} g(v_1)\,  I_{1,j}(V)\,dV
  &= \int\limits_{\mathclap{|V|^2+|W|^2=N+M}} g(v_1)
    \left(\barfnm(V_{w_j^1}^j,t)-\barfnm(V,t)\right)\,d\sigma(V,W) = 0\,,
\end{align*}
by the argument in equation~(\ref{eq:16b}), and therefore it is
sufficient to consider the
first term,  $I_{1,1}(V)$.
Using Equation~(\ref{nmmarginal}) and Definition~\ref{def:2}, the
integral in this term is 
\begin{align*}
  \int\limits_{\Omega_{N,0}}&
    \barfnm(V_{w_1}^1,t) d\sigma(W) - \left|\Omega_{N,0} \right| \barfnm(V,t)\\
  &=  \left|\Omega_{N,0} \right| \frac{1}{\s^{M-1}\left(\sqrt{M}\right)}
  \int\limits_{\s^{M-1}\left(\sqrt{M}\right)}
  \barfnm(V_{\tilde{w}_1 \sqrt{\tau(V)}}^1,t)\,d\sigma(\widetilde{W})
  - \left|\Omega_{N,0} \right| \barfnm(V,t)\,,
\end{align*}
where  we have made the change of variables $W\mapsto\sqrt{\tau(V)}
\widetilde{W}$, with 
\begin{align}
  \label{eq:taudef}
  \tau(V) = \frac{N+M-|V|^2}{M}\,,
\end{align}
the average energy per active particle. The integral is then the
marginal distribution of the uniform density over an
$M-1$-dimensional sphere, as in Equation~(\ref{eq:016a}), so that
finally the $I_{1,1}(V)$ becomes
\begin{align*}
  \begin{split}
 \Gamma_{N,0}  \left|\Omega_{N,0}\right|
           &\int\limits_{-\sqrt{M}}^{\sqrt{M}}
           \barfnm\left(V_{w \sqrt{\tau(V)}}^1,t\right) \Psi_M(w) \,dw -
          \Gamma_{N,0}  \left|\Omega_{N,0} \right| \barfnm(V,t) = \\
     &\int\limits_{-\sqrt{ N+M-|V|^2}}^{\sqrt{ N+M-|V|^2}}
     \frac{\Psi_M\left( v_1/ \sqrt{\tau(V^1_{w_1})}\right)}
     {\sqrt{\tau(V^1_{w_1})}} \fnmnnoll(V_{w_1}^1,t)dw_1 -
     \fnmnnoll(V,t)   
    \,.
    \end{split}
\end{align*}

This integral can now be written as a sum of three terms as follows:
\begin{equation}
\label{eq:threeterms}
\begin{split}
&\displaystyle\int\limits_{-\sqrt{\eact}}^{\sqrt{\eact}}
\mathcal{M}(v_1)\fnmnnoll(V_{w_1}^1,t)\ dw_1- \fnmnnoll(V,t)\\
&\hspace{1cm}+\displaystyle\int\limits_{-\sqrt{\eact}}^{\sqrt{\eact}}
\left(\mathcal{M}_M(v_1)-\mathcal{M}(v_1)\right)\fnmnnoll(V_{w_1}^1,t)\ dw_1\\
&\hspace{1cm}+\displaystyle\int\limits_{-\sqrt{\eact}}^{\sqrt{\eact}}
\left(\frac{\mathcal{M}_M\left(v_1/\sqrt{\tau(V_{w_1}^1)}\right)}
  {\sqrt{\tau(V_{w_1}^1)}}-\mathcal{M}_M(v_1)\right)\fnmnnoll(V_{w_1}^1,t)\  dw_1\\
&\hspace{0.45cm}:=a(V,t)+b(V,t)+c(V,t).
\end{split}
\end{equation}

Consider the first of these terms, $a(V,t)$. For arbitrary $g(v)\in C(\R)$, 
\begin{align}
  \label{eq:27_0}
  \int\limits_{|V|^2\le N+M}  a(V,t) g(v_1)\,dV
   &=  \int\limits_{v_1^2\le N+M} \left( \mathcal{M}(v_1) - 
    \fnmettnoll(v_1,t)\right)g(v_1)\,dv_1\,,  
\end{align}
which converges to the right-hand side of equation (\ref{eq:224}) in
Theorem~\ref{thm:main}.  Therefore  the proof can be concluded by
proving that the other terms vanish.

The second term, an integral of $b(V,t)$, converges to zero, because
\begin{equation}
  \label{eq:b11}
  \begin{split}
  \int\limits_{|V|^2\le N+M}   g(v_1) b(V,t)\,dV
  &= \int\limits_{|V|^2+w_1^2\le N+M} g(v_1)\left(
    \mathcal{M}_M(v_1)-\mathcal{M}(v_1) \right) \fnmnnoll(V_{w_1}^1,t)\,dw_1 dV\\
 &= \int\limits_{|V|^2\le N+M}
 \int\limits_{-\sqrt{N+M-|V|^2}}^{\sqrt{N+M-|V|^2}}
     \left(\mathcal{M}_M(u)-\mathcal{M}(u)\right) g(u) \, du \, \fnmnnoll(V,t)\,dV\,, 
\end{split}
\end{equation}
and we know that the $\mathcal{M}_M(w) \rightarrow \mathcal{M}(w)$ pointwise,
when $M\rightarrow\infty$.

Of the three terms $a(V,t), b(V,t)$, and $c(V,t)$, the last one is the
most difficult  
to analyse. This is the subject of Section~\ref{sec:moments}, where it
is proven that on the domain of integration, we have
$\tau(V)\rightarrow 1$ when $N,M\rightarrow\infty$ under the
constraints given in Theorem~\ref{thm:main}.

The proof of Theorem~\ref{thm:main} can be concluded with these three
estimates, together with a proof that $\int g(v_1)
I_2(V)\,dv\rightarrow 0$  when $N,M\rightarrow\infty$. The remainig
part of this section is devoted to  proving that $I_2(V)$
converges to zero with suitable choices of $N,M$ and $\lambda$.
The result is
largely due to Lemma~\ref{LemmaSpec} below, which in turn follows from
a result on the spectral gap for the generator $L_{NM\lambda}$ of
the Kac walk. Without the operator $U_{NM}$, the
generator for 
  the exchange between passive and active  
particles, equation~(\ref{mastereqVW}) becomes
\begin{align*}
  \frac{\partial}{\partial t} \fnm(V,W,t) &= L_{NM\lambda} \fnm(V,W,t)\,,
\end{align*}
which is the original Kac master equation in the $M$ variables
$(w_1,...,w_M)$, with the $V$-variables appearing only as
parameters. Kac conjectured that
the spectral gap $\Delta_M$
of $-\mathcal{L}_{M}=-(\lambda N)^{-1} L_{NM\lambda}$ is
bounded away from 
$0$, uniformly in the number of particles $M$. The parameter
$V$ is of course not present i Kac's work, but it doesn't have an
influence on the spectral gap, because this gap does not depend on the
total energy of the system, and $\lambda N$  only serves to increase
the jump rate.  Kac's conjecture was first proved by  
Janvresse, \cite{Janvresse2001}, and an exact formula for the  gap of
$-\mathcal{L}_{M}$ was 
later obtained by Carlen, Carvalho and Loss in
\cite{CarlenCarvalhoLoss2003}:
\begin{equation}
  \label{eq:specgap}
\Delta_M=\frac{1}{2}\frac{M+2}{M-1} \ge \frac{1}{2} \,, 
\end{equation}
and the corresponding eigenfunction is
\begin{align}
  \label{eq:eigfunc}
  \phi_{\Delta_M}&=\sum_{j=1}^M \left(   w_j^4- \frac{3M}{M+2}\left(\frac{r^2}{M}\right)^{2} \right)\,,
\end{align}
where $r$ is the radius of the $M-1$-dimensional sphere.
In fact, this eigenvalue and eigenfunction were computed also
in~\cite{Henin1974}, but without a proof that this is also determines
the spectral gap.

It follows that in an interval $t_1<t<t_2$ defined by two consecutive
exchange events,
\begin{equation}\label{spectralgap}
||\fnm(V,\cdot,t)-\barfnm(V,t) ||_2 \leq e^{-\frac{ N \lambda}{2}(t-t_1)}
||\fnm(V,\cdot,t_1)-\barfnm(V,t_1) ||_2\,.
\end{equation}
With $\lambda$ very large, the term $U_{NM}
F$ can be considered to be a small perturbation, which is expressed in
the following lemma: 

\begin{lemma}\label{LemmaSpec}
Let  $\fnm(V,W,t)$ be a solution to equation (\ref{mastereqVW}), and let
$\barfnm(V,t)$ be defined by Definition~\ref{def:2}. Then, for all $t\geq
0$ and $V\in \R^N$, 
\begin{equation}
\begin{split}
\norm{\fnm(V,\cdot,t)-\barfnm(V,t)}_2
     & \leq e^{-\frac{tN\lambda}{2}}\norm{\fnm(V,\cdot,0)-\barfnm(V,0)}_2\\
    &\hspace{0.3  cm}+\frac{1-e^{-\frac{tN\lambda}{2}}}{2N\lambda}
     \sup_{0\leq s\leq t}\norm{U_{NM}\fnm(V,\cdot,s)-\overline{U_{NM}F}(V,s)}_2,
\end{split}
\end{equation}
where the norm is in $L^2 \left(\Omega_{N,0},d\sigma(W) \right)$.
\end{lemma}

\begin{proof}
By the Duhamel formula we can write
\begin{equation}
  \label{eq:023}
\fnm(V,W,t)=e^{tL_{NM\lambda}}\fnm(V,W,0)+\int_{0}^{t}e^{(t-s)L_{NM\lambda}}U_{NM}\fnm(V,W,s)\ ds\,.
\end{equation}
Because $e^{tL_{NM\lambda}}$ acts only in the $W$ variables, and
conserves mass and leaves the uniform density on the sphere invariant,
an integration over $\Omega_{N,0}$ gives
\begin{align*}
\barfnm(V,t)=e^{tL_{NM\lambda}}\barfnm(V,0)+\int_{0}^{t}e^{(t-s)L_{NM\lambda}}
   \frac{1}{\left|\Omega_{N,0}\right|}
     \int\limits_{\Omega_{N,0}}U_{NM}\fnm(V,W,s)\, d\sigma(W) \, ds\,,
\end{align*}
and so
\begin{equation}
\begin{split}
   \fnm(V,W,t)-&\barfnm(V,t)=e^{tL_{NM\lambda}}\left(\fnm(V,W,0) - \barfnm(V,0)\right)  \\
     &+\int_{0}^{t}e^{(t-s)L_{NM\lambda}}\left(U_{NM}\fnm(V,W,s)-\overline{U_{NM}F}(V,W,s)\right)\ ds\,.
\end{split}
\end{equation}
The formula for the spectral gap for the Kac model yields
\begin{equation}
\begin{split}
   ||&\fnm(V,W,t)-\barfnm(V,t) ||_{L^2(\s^{M-1}(\sqrt{\eact}),d\sigma)}  \\
     & \leq e^{-\frac{tN\lambda}{2}}
        ||\fnm(V,W,0)-\barfnm(V,0)||_{L^2(\Omega_{N,0}),d\sigma)}\\
     &\hspace{0.3cm}+ \int_{0}^{t} e^{-\frac{(t-s)N\lambda}{2}}||U_{NM}\fnm(V,W,s)-\overline{U_{NM}F}(V,W,s)||_{L^2(\s^{M-1}(\Omega_{N,0}),d\sigma)}
     \ ds\,.
\end{split}
\end{equation}
A simple computation concludes the proof.
\end{proof}
The desired estimate of $I_2(V)$ is a direct  consequence of
Lemma~\ref{LemmaSpec}: 

\begin{lemma}
  \label{lem:I2estimate}
  Assume that $\fnm(V,W,0)\in L^2(\Omega_{0,0}, d\sigma(V,W))$. Let
  $\fnm(V,W,t)$ be the solution of equation~(\ref{mastereqVW})  and
  $\barfnm(V,t)$ be given by  Definition~\ref{def:2}
    (equation~(\ref{eq:def2eq})). Then, for every bounded function $g:
    \R \rightarrow \R$ and all $t\ge 0$
  \begin{equation}
    \label{eq:lem2.2eq}
    \left|\int\limits_{\,|V|^2\le \mathrlap{N+M}} I_2(V)\, g(v_1) \,dV
    \right|^2 \le 
    \left( 2 e^{- \lambda N t}   + \frac{4}{\lambda^2} \right)\| g
    \|_{\infty}^2 \| \fnm(\cdot,\cdot,0) \|_{L^2(\Omega_{0,0})}^2 
  \end{equation}
\end{lemma}

\begin{proof}
  Multiplying equation~(\ref{mastereqVW}) by $\fnm(V,W,t)$ and
  integrating over $\Omega_{0,0}$ with respect to $\sigma(V,W)$, and
  using that $L_{MN\lambda}$ is a non-positive operator, it follows
  that
  \begin{equation}
    \label{eq:x2.18}
    \frac{d}{dt}\int\limits_{\Omega_{0,0}} |\fnm(V,W,t)|^2\,d\sigma(V,W)
    \le \int\limits_{\Omega_{0,0}} \fnm(V,W,t)
    U_{NM}\fnm(V,W,t)\,d\sigma(V,W)\,. 
  \end{equation}
  According to the definition of $U_{NM}$~(equation~(\ref{U})),
the righthand side of this expression is a sum of terms of the form $
\fnm(V,W,t)\left(\fnm(V^j_{w_k},W^k_{v_j},t)-\fnm(V,W,t)\right)$, which due to
the inequality $2ab\le a^2+b^2$ is smaller than or equal to
$\frac{1}{2}\left(\fnm(V^j_{w_k},W^k_{v_j},t)^2-\fnm(V,W,t)^2\right)$. After
integration all these terms give a
non-positive contribution, and hence the $L^2$-norm of $F$ is non
increasing.  Therefore 
 the righthand side of the inequality~(\ref{eq:x2.18}) is
non-positive, and we have
\begin{equation}
  \label{eq:x18b}
  \int\limits_{\Omega_{0,0}} |\fnm(V,W,t)|^2\,d\sigma(V,W)\le
  \int\limits_{\Omega_{0,0}} |\fnm(V,W,0)|^2\,d\sigma(V,W)\,, 
\end{equation}
and  $\fnm(V,W,t)\in L^2(\Omega_{0,0})$ for all $t\ge0$. The
function $I_2(V)$ is defined as the second sum  in the right hand side
of equation~(\ref{eq:N0marg}). When multiplying that expression with a
function $g(v_1)$ (a function depending only on one variable), only
the term with $j=1$ remains, and therefore
\begin{equation}
  \label{eq:x18c}
  \begin{split}
    \int\limits_{|V|^2\le N+M} I_2(V) g(v_1) dV &=
    \int\limits_{|V|^2\le N+M} \Gamma_{N,0} g(v_1)
    \int\limits_{\Omega_{N,0}} \left(\fnm(V_{w_1}^1,W_{w_1}^1,t) -
      \barfnm(V_{w_1}^1,t)\right) \,d\sigma(W)\,dV \\
    & = \int\limits_{\Omega_{0,0}} g(v_1)  \left(\fnm(V_{w_1}^1,W_{w_1}^1,t) -
      \barfnm(V_{w_1}^1,t)\right)\,d\sigma(V,W)\,. \nonumber \\ 
  \end{split}
\end{equation}
Using the Cauchy-Schwartz inequality, and  Lemma~\ref{LemmaSpec} we get
\begin{equation}
  \label{eq:x18d}
  \begin{split}
    \left| \int\limits_{|V|^2\le N+M} I_2(V) g(v_1)\,dV \right|^2 &\le
    \| g\|_{\infty}^2 \int\limits_{|V|^2\le N+M} \Gamma_{N,0}
    \int\limits_{\Omega_{N,0}} \left| \fnm(V,\cdot,t) - \barfnm(V,t)
    \right|^2 \, d\sigma(W) dV \\
    &\hspace{-4em}\le 2 \| g\|_{\infty}^2 e^{-\lambda N t} \int\limits_{|V|^2\le
      N+M} \Gamma_{N,0} \int\limits_{\Omega_{N,0}} \left| \fnm(V,\cdot,0)
      - \barfnm(V,0) \right|^2 \, d\sigma(W) dV \\
    &\hspace{-4em} + \frac{2 \| g\|_{\infty}^2}{N^2 \lambda^2}
    \int\limits_{|V|^2\le N+M}\Gamma_{N,0} \sup\limits_{0\le s \le t}
    \int\limits_{\Omega_{N,0}} \left|U_{NM} \fnm(V,\cdot,s) -
      \overline{U_{NM}F}(V,s) \right|^2\,d\sigma(W)dV\,. \nonumber \\  
  \end{split}
\end{equation}
A calculation using~(\ref{U}) shows that
\begin{equation}
  \int\limits_{\Omega_{0,0}} \left|U_{NM} \fnm(V,\cdot,s) -
    \overline{U_{NM}F}(V,s) \right|^2\,d\sigma(W)dV\le
  2N^2 \int\limits_{\Omega_{0,0}} | \fnm(V,W,s)| ^2 \,d\sigma(V,W)\,,
\end{equation}
and so, collecting all the inequalities we finally obtain the
inequality~(\ref{eq:lem2.2eq}), which concludes the proof.
\end{proof}

\section{Evolution of moments and energy partition}
\label{sec:moments}

The average energy per active particle is
\begin{align*}
  \int_{\Omega_{0,0}} \fnm(V,W) \frac{1}{M}|W|^2\,d\sigma(V,W) =
  \int_{\Omega_{0,0}} \fnm(V,W) \frac{1}{M}\left(N+M-|V|^2\right)\,d\sigma(V,W)
\end{align*}
and the main purpose with this section is to prove that, for large
$N,M$,  the density $\fnm$ is concentrated near the set
$\left\{ 
  (V,W)\in\Omega_{0,0}\,\,\bigg|\,\,\displaystyle \frac{1}{M}|W|^2=1\right\} $. The proof
goes by estimates of moments of the form
\begin{align}
  \label{eq:30-1}
      \int_{\Omega_{0,0}} \fnm(V,W) H(V)\,d\sigma(V,W)\,\quad \mbox{or}\qquad
  \int_{\Omega_{0,0}} \fnm(V,W) H(W)\,d\sigma(V,W)\,,
\end{align}
one of them being
\begin{align}
  \label{eq:psidef0}
  \psi(t) &= \int_{\Omega_{0,0}} \fnm(V,W) \left(
            \frac{1}{M}|W|^2-1\right)^2 \,d\sigma(V,W)\,.
\end{align}
We prove that for all $t>0$, $\psi(t)\rightarrow 0$ when $N,M,
\lambda \rightarrow\infty$ in a suitable way.

The starting point is equation~(\ref{eq:023}),
\begin{align}
  \label{eq:3.001}
  F_t &= e^{t L} F_0+\int_{0}^t e^{(t-s)L} U F_s \,ds\,,
\end{align}
where, simplifying notation,  $F_t = \fnm(V,W,t)$, and the
operators $L= L_{N,M,\lambda}$, and $U=U_{N,M}$ are given
in  equation~(\ref{L}) and~(\ref{U}). The two operators $U$ and $L$ are
self adjoint, but they don't commute. Multiplying the terms in
(\ref{eq:3.001}) with $H(V)$ and integrating gives
\begin{align}
  \label{eq:3.010}
  \int_{\Omega_{0,0}} F_t(V,W)&\,H(V)\,d\sigma \nonumber \\
   &=  \int_{\Omega_{0,0}}\left(e^{t L}  F_0(V,W)\right)\,H(V)\,d\sigma +
       \int_{0}^t\int_{\Omega_{0,0}}\left(e^{(t-s) L} \left[U F_s\right]
                                              (V,W)\right) H(V)\,d\sigma
                                              \,ds \nonumber\\
   &= \int_{\Omega_{0,0}}F_0(V,W)e^{t L} H(V)\,d\sigma +
       \int_{0}^t\int_{\Omega_{0,0}} \left[U F_s\right] (V,W) e^{(t-s) L}
  H(V)\,d\sigma \,ds\\
   &= \int_{\Omega_{0,0}}F_0(V,W) H(V)\,d\sigma +
       \int_{0}^t\int_{\Omega_{0,0}} \left[U F_s\right] (V,W)
     H(V)\,d\sigma \,ds\nonumber\\
    &= \int_{\Omega_{0,0}}F_0(V,W) H(V)\,d\sigma +
       \int_{0}^t\int_{\Omega_{0,0}} F_s (V,W)
 \left[U H\right](V,W)\,d\sigma \,ds\nonumber
\end{align}
because $L$ acts only in the $W$-variables, and constants are left
invariant by $L$.

The calculations for moments of the form $H(W)$, are much simplified
 when $H(W)$ is an eigenfunction of the operator $L$ with eigenvalue
 $\lambda_H$. In this case
\begin{align}
  \label{eq:3.015}
  \int_{\Omega_{0,0}} F_t(V,W)&\,H(W)\,d\sigma \nonumber \\
   &= \int_{\Omega_{0,0}}F_0(V,W)e^{t L} H(W)\,d\sigma +
       \int_{0}^t\int_{\Omega_{0,0}} \left[U F_s\right] (V,W) e^{(t-s) L}
  H(W)\,d\sigma \,ds\\
    &= e^{\lambda_H t}\int_{\Omega_{0,0}}F_0(V,W) H(W)\,d\sigma +
      \int_{0}^t e^{(t-s)\lambda_H}   \int_{\Omega_{0,0}} F_s (V,W)
 \left[U H\right](V,W)\,d\sigma \,ds\,.\nonumber
\end{align}
In all the integrals we need to compute $[UH](V,W)$, where the
functions $H(V)$ and $H(W)$ are given by expressions of the form
\begin{equation}
  \label{eq:30-2}
  \begin{split}
  H(V) &= \frac{1}{N}\sum_{j=1}^N h(v_j)\,,\;\; \mbox{or} \\
  H(V) &= \Phi( |V|^2)\,,
\end{split}
\end{equation}
or combinations of these, and similar with functions depending only on
$W$. Then $[UH](V,W)$ is given by the sum in~(\ref{U}) as
\begin{align}
  \label{U-2}
  \frac{1}{M}\sum_{j=1}^N \sum_{k=1}^M \left( H(V_{w_k}^j)-H(V)\right)\,,
\;\mbox{or}\;
  \frac{1}{M}\sum_{j=1}^N \sum_{k=1}^M \left( H(W_{v_j}^k)-H(W)\right)\,.
\end{align}
When $H(V)$ and $H(W)$ are of the form~(\ref{eq:30-2}), then
\begin{align}
  H(V_{w_k}^j) =  H(V) + \frac{1}{N}\left( h(w_k) -h(v_j)\right) \,,\;\mbox{and}
  \qquad H(V_{v_k}^j) =  \Phi( |V|^2 + w_k^2-v_j^2)
\end{align}
respectively. 
 
The mean energy per active particle in a given configuration $(V,W)$ is
\begin{align}
  \tau(V) = \frac{|W|^2}{M}=\frac{1}{M}\left( N+M-|V|^2\right) = 1 +
  \frac{N-|V|^2}{M}\,, 
\end{align}
and  in addition to $\tau(V)$ we  introduce the notation  
\begin{align}
  m_4(W)& =\frac{1}{M}\sum_{k=1}^M w_k^4 \qquad\mbox{and}\qquad 
  m_4(V) =\frac{1}{N}\sum_{j=1}^N v_j^4\,,
\end{align}
as well as
\begin{equation}
  \widetilde{m}_4(V) = \frac{1}{|\Omega_{N,0}|} \int_{\Omega_{N,0}}
  m_4(W)\,d\sigma(W)\,, 
\end{equation}
 the average of $m_4(W)$ over the sphere $\Omega_{N,0}$\,. This average can be expressed in terms of the radius of
$\Omega_{N,0}$, $|W|$, which in turn is a function of $|V|$. We have
\begin{align}
  \begin{split}
  \widetilde{m}_4(V) &= \frac{3 M}{M+2}\left(\frac{|W|^2}{M}\right)^2=
  \frac{3M}{M+2}\left(\frac{N+M-|V|^2}{M}\right)^2 \\
  &= \frac{3 M}{M+2}\left((\tau(V)-1)^2+ 2(\tau(V)-1) +1\right)\,.     
  \end{split}
\end{align}
Using this notation, we define the following moments:
\begin{equation}
  \label{eq:3.etadef}
  \begin{split}
  \eta(t)&=
     \int_{\Omega_{0,0}} \fnm(V,W,t)   
     \left( \tau(V)-1 \right)\,d\sigma\,,\\
  \psi(t)&=
  \int_{\Omega_{0,0}} \fnm(V,W,t) \left( \tau(V)-1 \right)^2
  \,d\sigma\,,\qquad\mbox{(already defined in
    eq.~(\ref{eq:psidef0}))}\\ 
   \zeta(t) &= \int_{\Omega_{0,0}}
   \fnm(V,W,t)\left(m_4(W)-\widetilde{m}_4(V)\right)
   \,d\sigma\,,\quad\mbox{and} \\ 
   \xi(t)&=
   \int_{\Omega_{0,0}} \fnm(V,W,t)m_4(V) \,d\sigma\,.
   \end{split}
\end{equation}
The expression  $m_4(W)-\widetilde{m}_4(V)$ is the eigenfunction
of the operator $L$ corresponding to the eigenvalue $\lambda N \Delta_M$.
We now obtain expressions for these moments using the equations
(\ref{eq:3.010}) and (\ref{eq:3.015})\footnote{Some of the
  calculations are rather messy, and we have used a
  computer algebra system for checking these calulations as well
  as for analysing the linear system for the moments. The
  code showing all operations are available upon request from the
  corresponding author.}. 

First setting 
$H(V)=\tau(V)-1=\frac{N-|V|^2}{M} $ gives 
\begin{align*}
  -\frac{N}{M}\left( \frac{1}{M}|W|^2-\frac{1}{N}|V|^2\right)
  &= -\frac{N}{M}\left( \frac{1}{M}\left(N+M-|V|^2\right)-\frac{1}{N}|V|^2\right)\\
  &= -\left(1+\frac{N}{M}\right)
      \left( \tau(V)-1 \right)\,,
\end{align*}
and therefore, with $\nu_{N,M}= 1+N/M$,
\begin{align}
  \eta(t) &= \eta(0) -\nu_{N,M}  \int_{s=0}^t \eta(s)\,ds\,,
\end{align}
which means that the mean energy per active particle converges
exponentially  to $1$ as $N/M\rightarrow\infty$\,.

Similarly, to find an expression for $\psi(t)$, we take 
  $H(V) = \left(\tau(V)-1\right)^2 =
  \left(\frac{N-|V|^2}{M}\right)^2$ and  obtain the expression 
\begin{align}
  \label{eq:3.040}
H(V_{w_k}^j)-H(V) &=
                    \left(\frac{N-|V|^2+v_j^2-w_k^2}{M}\right)^2-
                    \left(\frac{N-|V|^2}{M}\right)^2  \nonumber \\
  &= \frac{2}{M} \left(\frac{N-|V|^2}{M}\right) \left(v_j^2-w_k^2\right) +
    \left(\frac{v_j^2-w_k^2}{M}\right)^2\,. 
\end{align}
Using
\begin{align}
  \begin{split}
  &\frac{1}{M}\sum_{j,k} (v_j^2-w_k^2)  =-\nu_{N,M} M (\tau(V)-1)\qquad  \mbox{and}\\
  &\frac{1}{M} \sum_{j,k} v_j^2 w_k^2=\frac{1}{M} |V|^2|W|^2=
  -M (\tau(V)-1)^2+ (N-M)(\tau(V)-1) +N
  \end{split}
\end{align}
and evaluating the sum in~(\ref{U-2}) gives
\begin{align}
    \label{eq:3.050}
  \left[U H\right]&(V,W) =-2\left(\nu_{N,M}-\frac{1}{M}\right)
    \left(\tau(V)-1\right)^2
     -\left(\frac{2N}{M^2}-\frac{2}{M}\right)
                    \left(\tau(V)-1\right) -\frac{2N}{M^2}\\
                  &\qquad\qquad + \frac{N}{M^2}
                    \left( m_4(V)+m_4(W)\right)\,. \nonumber
\end{align}
Because
\begin{align}
  m_4(W) &= m_4(W) - \widetilde{m}_4(V) + \frac{3M}{2+M}\left( (\tau(V)-1)^2
           + 2(\tau(V)-1) +1\right) 
\end{align}
it follows that  $\psi(t)$ satisfies
\begin{align}
  \label{eq:3.050c}
  \psi(t) &= \psi(0) +
   \int_{0}^t \left( -2\left(\nu_{N,M} -\frac{1}{M}\right) \psi(s)
              - \left(\frac{2N}{M^2}-\frac{2}{M}\right)\eta(s)
              -\frac{2N}{M^2} \right)\,ds \nonumber \\
            &+ \int_0^t \frac{N}{M^2}\left( \xi(s) + \zeta(s) +
              \frac{3 M}{M+2}\left( \psi(s) + 2\eta(s) +1 \right)\right)\,ds\,.
\end{align}

Next, taking  $H(V)=m_4(V)$ in~(\ref{eq:3.010}) gives an expression for
$\xi(t)$. With  
\begin{align}
  [UH](V,W) &= m_4(W)-m_4(V) = -m_4(V) +
              \left(m_4(W)-\bar{m}_4(V)\right) +\frac{3M}{M+2} \tau(V)^2\,,
\end{align}
and integrating over $V$ and $W$ results in
\begin{align}
  \xi(t) &= \xi(0) + \nonumber \\
         &\quad\int_0^t \left( -\xi(s) + \zeta(s)
           +\frac{3M}{M+2}\left( \psi(s) +2\eta(s) +1 \right)
           \right)\,ds\,. 
\end{align}

And finally, with $H(W)=m_4(W)-\widetilde{m}_4(V) = \frac{1}{M}\sum_{k=1}^M w_k^4
-\frac{3}{M(M+2)} |W|^4$,   we get
\begin{align}
  H\left(W_{v_j}^k\right)-H(W) &= \frac{1}{M}\left(v_j^4-w_k^4\right)
                                 -  \nonumber\\
  &\quad\quad \frac{3}{M(M+2)}\left(2 |W|^2\left( v_j^2-w_k^2\right) + v_j^4+
    w_k^4 -2v_j^2 w_k^2\right)\,,
\end{align}
and summing over $j$ and $k$ as before
\begin{align}
  \begin{split}
  [UH](V,W) =& \frac{N}{M}\left(1-\frac{3}{M+2}\right)m_4(V)
               - \frac{N}{M}\left(1+\frac{3}{M+2}\right)m_4(W) \\
             &  -\frac{6}{M(M+2)}\left(1-\frac{1}{M}\right)|W|^2|V|^2 +
               \frac{6N}{M^2(M+2)}|W|^4\,.
\end{split}
\end{align}
After some manipulations, one then gets an equation for  $\zeta(t)$:
\begin{align}
  \label{eq:zeta}
  \begin{split}
    \\ 
    \zeta(t) =& e^{-N\lambda \Delta_M t} \zeta(0) +  \\
    &\int_0^t  e^{-N\lambda \Delta_M (t-s)}
    \left( \frac{N(M-1)}{M(M+2)}  \xi(s)
      -\frac{N(M+5)}{M(M+2)}\zeta(s) \right.\\
    &\qquad\qquad\qquad\qquad - \frac{6(M-1)(2N-M^2-2 M)}{M(M+2)^2}\eta(s) \\
    &\qquad\qquad\qquad\qquad\qquad + \left.\frac{3(N+2M+4)(M-1)}{(M+2)^2}\psi(s)+ \frac{3 N (4  - 3 M  - M^2)}{M (2 + M)^2}\right)\,ds\,.
  \end{split}
\end{align}
Differentiating these expressions, we find a linear system of
differential equations for $\Psi(t)=(\psi(t),\xi(t),\zeta(t))^{t}$,
\begin{align}
  \label{eq:diffsyst}
  \frac{d}{dt}\Psi(t) &= \mathbf{A} \Psi(t)  + {\mathbf b_1} +
                        {\mathbf b_2}\eta(t)\,.  
\end{align}
The initial values depend on the moments of the initial density
$\fnm(V,W,0)$, and are bounded by
\begin{align}
  \label{eq:diffinitvalues}
  |\eta(0)| &\le N/M\,, \nonumber \\
  0\le \psi(0) &\le (N/M)^2\,,\nonumber \\
  |\zeta(0)|&\le N^2/M,\qquad\mbox{and}\\
  0\le \xi(0)& \le 2 N\,. \nonumber
\end{align}
These bounds may be achieved, and
hence the momentes are not bounded uniformly in $N$ and $M$ unless
further hypothesis are made on intial data. As we shall see, this is
not very critical  for  $\eta(t)$, $\psi(t)$, and $\zeta(t)$, because
after an initial interval of length 
$\sim M/N$, the transient part of the solution will be small for these
variables. However, 
this is not the case for $\xi(t)$, and to conclude our proof of
convergence to the BGK-equation, we shall have to assume the
initial data $\fnm(V,W,0)$ are such that $\xi(0)/M\rightarrow 0$
when $M\rightarrow\infty$. 

We continue with an asymptotic analysis of the system of
equations~(\ref{eq:diffsyst}). 
The matrix $\mathbf A$ and vectors ${\mathbf b}_1$  and ${\mathbf
  b}_1$ are expressions 
involving $N$ and $M$.  We have ${\mathbf A} = {\mathbf A^0} + {\mathbf A^r}$,
where 
\begin{align}
  {\mathbf A^0} &= \begin{pmatrix}
                   -2\frac{N}{M}  
                    & \frac{N}{M^2} & \frac{N}{M^2}\\
                     3 \frac{N}{M}  
                     &
                     -\frac{N}{M}-N \lambda \Delta_M
                       &
                          \frac{N}{M}  
                        \\
                       3  
                       &
                       1
                      & -1
    \end{pmatrix}\,,
\end{align}
where in the third element of the third row, we have taken the
exponential factor 
in~(\ref{eq:zeta}) into account, and the components of ${\mathbf A^0}$
and ${\mathbf A^r}$ satisfy ${\mathbf a^r}_{ij} /{\mathbf a^0}_{ij}=
\bigoh\left(M/N + 1/M \right)$, or smaller, when
$N,M\rightarrow\infty$.  And similarly for the constant vector
${\mathbf b}= {\mathbf b^0}+{\mathbf b^r}$, we get 
\begin{align}
  \label{eq:b1}
  {\mathbf b^0_1} = \begin{pmatrix}
 \frac{N}{M^2} 
                        \\
                        -3\frac{N}{M} 
                         \\
                          3 
                        \end{pmatrix} \qquad \mbox{and}& \qquad
  {\mathbf b^0_2} = \begin{pmatrix}
    4 \frac{N}{M^2}  
    \\
    6-12 \frac{N}{M^2}  
     \\
       6  
     \end{pmatrix}\,.
\end{align}

The components of ${\mathbf b^r}$ are dominated by the components
of ${\mathbf b^0}$ in the same way as the elements of $\mathbf A$. And
because $\eta(t) \le N M^{-1} e^{-\frac{N}{M}t}$, we see that for any fixed
$t_0>0$, ${\mathbf b_2} \eta(t)$ will be negligible compared to
${\mathbf b}_1$, uniformly in $t>t_0$, when $N/M$ is sufficiently
large. For large values of $N, M$, and $\lambda$, all but the last line of
$\mathbf{A}^0$ will be dominated by the diagonal elements. The
eigenvalues of $\mathbf{A}$ are asymptotically

\begin{align}
  \ell_1 &= - 2\frac{N}{M}\left(1 +
           \bigoh\left(\frac{1}{M}\right) \right)\,,\\
  \ell_2
         &=-N\lambda\Delta_M
           \left(1 +
           \bigoh\left(\frac{1}{M}\right) \right)\, \qquad \mbox{and}
           \nonumber  \\   
  \ell_3 &= -1 + \bigoh\left(\frac{1}{M}\right)\nonumber\,,
\end{align}
and we set $\ell_0=-(1+N/M)$, the decay rate of $\eta(t)$.
The $\bigoh$-terms also contain terms of the form $\bigoh(1/N)$ and
$\bigoh(1/\lambda\Delta_M)$ but we will need $ \lambda\Delta_M >> N >>
M$, and therefore all remainder terms can be absorbed in one term that
is $\bigoh(1/M)$.

The solution to Equation~(\ref{eq:diffsyst}) is explicitly given by
\begin{align}
  \label{eq:diffsystSol}
  \Psi(t)  =\exp{(t {\mathbf A}) }\Psi(0)
    +
  \int_0^t   \exp{((t-s){\mathbf A}))} \left({\mathbf b_1} + {\mathbf b_2}\eta(s)\right)\,ds \,.
\end{align}

We use Sylvester's formula to compute the exponential $\exp(tA)$:

\begin{align}
  \label{eq:sylvester}
  e^{t\mathbf{A}} &= \sum_{j=1}^3 e^{\ell_j t} \mathbf{A}_j\,,
\end{align}
where
\begin{align}
  \label{eq:sylvmatriser}
  \mathbf{A}_j&= \prod_{k, k\ne j} \frac{1}{\ell_j-\ell_k}\left(
                \mathbf{A}-\ell_k I\right)\,.
\end{align}
Therefore
\begin{align}
  \Psi(t) &= \sum_{j=1}^3 
            \left(
            e^{\ell_j t}\mathbf{A}_j \Psi(0) +\frac{1}{\ell_j}(1-e^{\ell_j t})
            \mathbf{A}_j \mathbf{b}_1 +\frac{1}{\ell_0-\ell_j}\left(e^{\ell_0
            t}-e^{\ell_j t}\right) \mathbf{A}_j \mathbf{b}_2\eta(0)
            \right)\,,
\end{align}
and this expression can be evaluated at least asymptotically as $N/M$,
$M$, and $\lambda$ increase to infinity.
For the purpose of this paper we need to prove that for any $t>0$ (and
uniformly for $t\ge t_0>0$), $\psi(t)\rightarrow 0$ when $N/M$, $M$, and
$\lambda\rightarrow\infty$ as stated below, but all components of
$\Psi(t)$ are needed to obtain a closed system.

The components of the matrices are all rational expressions of $N/M$,
$M$, and $\lambda\Delta_M$, therefore all terms involving a  $e^{\ell_j
  t},\; j=0,1,2$ vanish when $M$, $N/M$ and $\lambda \Delta_M$ increase to
infinity and for $t>0$, uniformly for  $t\ge t_0>0$. It is therefore
enough to study the terms that are constant in $t$ or with an
exponential factor $e^{\ell_3 t}$ (we recall that $\ell_3\sim -1 $ in
the limit of interest). 

The matrices $\mathbf{A}_j,\; j=1,2,3$ are,  asymptotically,
\begin{align}
  & \mathbf{A}_1=
  \begin{pmatrix}
   1 & \frac{1}{M^2\lambda\Delta_M}  & -\frac{1}{2M} \\[0.8em]
    \frac{3}{M\lambda\Delta_M} & -\frac{9}{4M}   & \frac{9}{4M^2 \lambda
      \Delta_M } \\[0.8em] 
    -\frac{3 M}{2 N} & \frac{3}{4 M N \lambda \Delta_M } & \frac{3}{4N}        
  \end{pmatrix}
   \qquad\qquad
  \mathbf{A}_2=  \begin{pmatrix}
    \frac{3}{M^3 \lambda^2 \Delta_M^2} & -\frac{1}{M^2 \lambda \Delta_M }
    & \frac{1}{M^3 \lambda^2 \Delta_M^2} \\[0.8em]
    -\frac{3}{M\lambda\Delta_M} & 1   & -\frac{1}{M \lambda \Delta_M} \\[0.8em]
    \frac{3}{M N \lambda^2 \Delta_M^2} & -\frac{1}{N\lambda\Delta_M} &
    \frac{1}{M N \lambda^2\Delta_M^2}        
  \end{pmatrix} \nonumber  \\
  & \mbox{and} \qquad\qquad
  \\
  \nonumber
  &  \mathbf{A}_3 =  \begin{pmatrix}
    -\frac{3}{4 M^2} & -\frac{9}{4 M^3 \lambda \Delta_M}  & \frac{1}{2M} \\[0.8em]
    \frac{24}{M^2} & \frac{9}{4M}   & \frac{1}{M\lambda\Delta_M} \\[0.8em]
    \frac{3M}{2N} & \frac{1}{N\lambda\Delta_M}& 1        
  \end{pmatrix}\,.
\end{align}

The constant terms, which also absorb the exponential factors
which are related to the term $\mathbf{b}_1$ in
equation~(\ref{eq:diffsyst}), are 
$\frac{1}{\ell_j}\mathbf{A}_j \mathbf{b}_1$, 
and asymptotically these expressions have magnitudes
\begin{align}
  j=1:\quad & \frac{M}{2N}\left( \frac{N}{M^2}+
              \frac{1}{M^2\lambda\Delta_M}\frac{3 N}{M}+3 \frac{1}{2
              M}\right)\; \le \; \frac{C}{M} \nonumber \\
  j=2:\quad & \frac{1}{N\lambda\Delta_M}
              \left(\frac{3}{M^3 \lambda^2\Delta_M^2}
              \frac{N}{M^2}+\frac{1}{M^2\lambda\Delta_M} \frac{3 N}{M}+
              3 \frac{1}{M^3\lambda^2\Delta_M^2}      \right)\; \le \;
              \frac{C}{M^3\lambda^2\Delta_M^2} \\
  j=3:\quad & \frac{3}{4M^2}\frac{N}{M^2}+ \frac{9}{4
              M^3\lambda\Delta_M}\frac{3 N}{M}+ 3\frac{1}{2M}\; \le \;
              C \left( \frac{N}{M^4}+\frac{1}{M}  \right) \,. \nonumber
\end{align}

It remains to look at the exponential terms multiplying the initial
data. The relevant term is bounded by
\begin{align}
  \label{eq:estimate1}
 \frac{3}{4 M^2}\psi(0) + \frac{9}{4 M^3 \lambda \Delta_M}\zeta(0) +
  \frac{1}{2M}\xi(0)\,.
\end{align}
Already the constraints on the initial data stated in~(\ref{eq:diffinitvalues})
imply that the three terms are bounded by a constant multiplied by
\begin{align}
   \frac{ N^2}{M^4}, \qquad 
   \frac{N^2}{M^4\lambda\Delta_M},\qquad\mbox{and}\qquad \frac{ N}{M}\,,
\end{align}
respectively, and imposing that $N/M^2\rightarrow 0 $ when
$N\rightarrow\infty$ the first two terms of~(\ref{eq:estimate1})
vanish in the limit. The condition that $N/M^2\rightarrow\infty$ may
be relaxed to $N/M^4\rightarrow0$ if instead we require
$\psi(0)/M^2\rightarrow 0$ when 
$M\rightarrow 0$.  And  we do need to impose that
$\xi(0)/M\rightarrow 0$ when $N,M\rightarrow\infty$.  

Summarizing these estimates we obtain the following result:

\begin{lemma}
  \label{lem:moments}
  Let   $\eta(t), \psi(t), \zeta(t)$, and $\xi(t)$ be moments of
  solutions $\fnm(V,W,t)$ to equation~(\ref{mastereqVW}), as
  defined in equation~(\ref{eq:3.etadef}). Directly from the
  definition it follows that $\eta, \psi,\zeta$, and $\xi$ are
  bounded by $N/M, (N/M)^2, N^2/M$, and $2 N$, respectively. Assume
  that $N$, $M$ and $\lambda$ increase to infinity in such a way
  that $N/M\rightarrow\infty$, $N/M^2\rightarrow 0$,  and assume in
  addition that 
  \begin{align}
    \frac{\xi(0)}{M}  = \frac{1}{M}\int_{\Omega_{0,0}} \fnm(V,W,0)
    \left(\frac{1}{M}\sum_{k=1}^M w_k^4 \, \right)
    d\sigma \rightarrow 0.
  \end{align}
  Fix $t_0>0$. Then for $t\ge t_0$ there is a constant depending on
  $t_0$ such that
  \begin{align}
    \psi(t) =  \int_{\Omega_{0,0}} \fnm(V,W,0) \left(
    \tau(V)-1 \right)^2
    d\sigma \le C  \left( \frac{N^2}{M^4}  +\frac{\xi(0)}{M} \,\right)
  \end{align}
 when $N,M$, and $\lambda$ go to infinity.
  
\end{lemma}

\begin{remark}
  The lemma states that the mean energy per active particle converges
  to one when the number of particles increase as stated in the
  lemma. As presented here this is only certain for strictly positive
  times, but with further assumptions on the initial data $\fnm(V,W,0)$
  the same result could be achieved uniformly in time.
\end{remark}

We are ready to prove the main result of this section, which says that
when  $N$, $M$, and $\lambda$ go to infinity as in the previous
lemma, the first marginal of 
\begin{equation}
c(V)=\int\limits_{\R}\left(
\frac{1}{\sqrt{\tau(V_{w_1}^j)}}\mathcal{M}_M\left(\frac{v_j}{\sqrt{\tau(V_{w_1}^j)}}\right)-\mathcal{M}_M(v_j)\right) \fnmnnoll(V^j_{w_1},t)\ dw_1 
\end{equation}
from equation~(\ref{eq:threeterms}) converges to zero.

\begin{lemma}\label{thmfhat}
For all bounded functions  $g\in C(\R)$ , 
\begin{equation}
\left|\int_{\R^N}c(V) g(v_1)\ dV \right|\leq \sqrt{\psi(t)}\,,
\end{equation}
and this converges to zero when $N, M$, and $\lambda$ increase to
infinity according to Lemma~\ref{lem:moments}.
\end{lemma}

\begin{proof}
Multiplying $c(V)$ with $g(v_1)$ and integrating gives 
\begin{equation}\label{TVg1}
\begin{split}
&\int_{\R^N}c(V) g(v_1)\ dV=\\
&\int_{\R^N}\displaystyle\int\limits_{\R}\left(
\frac{1}{\sqrt{\tau(V_{w_1}^1)}}
\mathcal{M}_M\left(\frac{v_1}{\sqrt{\tau(V_{w_1}^1)}}\right)-\mathcal{M}_M(v_j)\right)
\fnmnnoll(V^1_{w_1},t)\, dw_1 g(v_1)\ dV \\
=& \int_{\R^N}\displaystyle\int\limits_{\R}
\left(\frac{1}{\sqrt{\tau(V)}}
  \mathcal{M}_M\left(\frac{u}{\sqrt{\tau(V)}}\right)-\mathcal{M}_M(u)\right) 
g(u)\, du  \, \fnmnnoll(V,t) \, dV\,.
\end{split}
\end{equation}

This integral may be estimated by noting that
\begin{align*}
  \begin{split}
  \left|\frac{1}{s}\mathcal{M}_M\left(\frac{u}{s}\right) - \mathcal{M}_M(u)\right|
  \le& \one_{|s^2-1|\ge \frac{1}{2}} \left(
  \frac{1}{s}\mathcal{M}_M\left(\frac{u}{s}\right) + \mathcal{M}_M(u)    \right)
\frac{|s-1|}{1/2}\\
&+ \one_{|s-1| <  \frac{1}{2}} \sup\limits_{|\bar{s}-1|<\frac{1}{2}}
\left|
  \frac{1}{\bar{s}^2}\mathcal{M}_M\left(\frac{u}{\bar{s}} \right) +
  \frac{u}{\bar{s}^3} \mathcal{M}_M'\left(\frac{u}{\bar{s}} \right)
   \right||s-1| 
\end{split}
\end{align*}

We have $\mathcal{M}_M(x)\rightarrow \frac{1}{\sqrt{2\pi}}
  \exp\left(-x^2/2\right)$ when $M\rightarrow\infty$ (see equation~(\ref{eq:016a})), and
similarly  $\mathcal{M}_M'(x)\rightarrow -  \frac{x}{\sqrt{2\pi}}
  \exp\left(-x^2/2\right)$, and therefore 
  that 
  \begin{align}
    \label{eq:81aa}
  \begin{split}
  \left|\frac{1}{s}\mathcal{M}_M\left(\frac{u}{s}\right) - \mathcal{M}_M(u)\right|
  \le&  \left( \frac{2}{s} e^{-(u/s)^2/2} + 2 e^{-u^2/2} +32  (1+u^2)e^{-
         u^2/8} \right) |s-1|
\end{split}
\end{align}
when $M$ is large enough. In equation~(\ref{TVg1}) we may then
estimate $g(v_1)$ with $\|g\|_{\infty}$, and then carry out the
integral over $u$ to get
\begin{align}
  \begin{split}
  \left|\int_{\R^N} c(V) g(v_1)dV \right| & \le C \|g\|_{\infty} \int_{\R^N}
    \fnmnnoll(V)   | \sqrt{\tau(V)}-1|\,dV \\
    & \le C \|g\|_{\infty} \int_{\R^N} \fnmnnoll(V)  |\tau(V))-1|  \,dV\\
     & \le C \|g\|_{\infty} \left(  \int_{\R^N}  \fnmnnoll(V)
       (\tau(V))-1)^2\,dV \right)^{1/2} = C \|g\|_{\infty} \sqrt{\psi(t)}
    \end{split}
\end{align}
The constant $C$ here is obtained by integrating the
expression~(\ref{eq:81aa}). Then Lemma~\ref{lem:moments} provides the
needed bounds for $\psi(t)$ 

\end{proof}

\section{Proof and conclusions}
\label{sec:final}

The goal of this section is to prove Theorem~\ref{thm:main}.
Recall from equation~(\ref{eq:threeterms}) that the equation for the
$(N,0)$-marginals can be written 
\begin{equation}
\frac{\partial}{\partial t} \fnmnnoll(V,t)=a(V)+b(V)+c(V)+ I_2(V),
\end{equation}
and hence that, for any bounded, continuous function $g(v_1)$, 
\begin{equation}
  \begin{split}
     \frac{\partial}{\partial t} &\int\limits_{|V|^2\le N+M}
     \fnmnnoll(V,t) g(v_1)\,dV =
\frac{\partial}{\partial t}
\int\limits_{\mathclap{-\sqrt{N+M}}}^{\mathclap{\sqrt{N+M}}} 
     \fnmettnoll(v_1,t) g(v_1)\,dv_1 = 
     \\ 
  & \int\limits_{\mathclap{|V|^2\le N+M}}  a(V) g(v_1)\, dV +
  \int\limits_{\mathclap{|V|^2\le N+M}} b(V) g(v_1)\, dV +
  \int\limits_{\mathclap{|V|^2\le N+M}}  c(V) g(v_1)\, dVa +   
  \int\limits_{\mathclap{|V|^2\le N+M}}  I_2(V) g(v_1)\, dV\,,
    \end{split}
\end{equation}
 These terms have been analysed above in this paper, and it only remains to put the pieces 
 together. The first term, the integral of $a(V)$,
\begin{align}
  \label{eq:27_0X}
  \int\limits_{|V|^2\le N+M}  a(V) g(v_1)\,dV
   &=  \int\limits_{v_1^2\le N+M} \left( \mathcal{M}(v_1) - 
    f_{N,M}^{1,0}(v_1)\right)g(v_1)\,dv_1\,.  
\end{align}
 converges to the
righthand side of equation~(\ref{eq:224}) ( see eq.~(\ref{eq:27_0})),

The second term, the integral of $b(V)$, converges to zero because
$\mathcal{M}_M(w)$ converges pointwise to the Maxwellian $\mathcal{M}$ (see
eq.~(\ref{eq:b11})).

That the third term converges to zero is exactly the content of
Lemma~\ref{thmfhat}, and finally Lemma~\ref{lem:I2estimate} states that
\begin{equation*}
  \int\limits_{\mathclap{|V|^2\le N+M} } I_2(V)\,  g(v_1)\,dV
\end{equation*}
converges to zero if $\| \fnm(\cdot,\cdot,0)
\|_{L^2(\Omega_{0,0})}/\lambda\rightarrow 0 $. By hypothesis (see
equation~(\ref{eq:223})), $ \| \fnm(\cdot,\cdot,0)
\|_{L^2(\Omega_{0,0})}\le C_{N M}$ for a family of constansts $C_{N
  M}$ and we may the choose $\lambda=\lambda_{N M}$ accordingly.
.

Now let $f(v,t)$ be the solution of
\begin{align}
  \partial_t f(v,t) &= \mathcal{M} -f(v,t)\,.
\end{align}
Then
\begin{align}
  \begin{split}
  \frac{\partial}{\partial t}  \int_{\R} g(v_1)& \left( \fnmettnoll
  (v_1,t)-f(v_1,t)\right)\,dv_1 = -\int_{\R} g(v_1) \left( \fnmettnoll
  (v_1,t)-f(v_1,t)\right)\,dv_1  \\
  & +  \int\limits_{\mathclap{|V|^2\le N+M}}  b(V)\, g(v_1)\, dV +
  \int\limits_{\mathclap{|V|^2\le N+M}}  c(V)\, g(v_1)\, dV +   
  \int\limits_{\mathclap{|V|^2\le N+M}}  I_2(V)\, g(v_1)\, dV\,,
  \end{split}
\end{align}
and it follows that
\begin{align}
   \int_{\R} g(v_1) \left( \fnmettnoll
  (v_1,t)-f(v_1,t)\right)\,dv_1
\end{align}
converges to zero under the assumptions of the theorem.

This concludes the proof of Theorem~\ref{thm:main}.

\section*{Acknowledgements}
We would like to thank Pierre Degond for helpful discussions. D.M. acknowledges support by the Swedish Science Council.
B.W. acknowledges support by the Swedish Science Council, the Knut and
Alice Wallenberg foundation and the Swedish Foundation for Strategic Research.

\bibliographystyle{plain}

\def\cprime{$'$} \def\cprime{$'$}
  \def\polhk#1{\setbox0=\hbox{#1}{\ooalign{\hidewidth
  \lower1.5ex\hbox{`}\hidewidth\crcr\unhbox0}}}

\end{document}